\documentclass[pageno]{jpaper}

\usepackage{algorithmicx}
\usepackage{algorithm,algpseudocode}
\usepackage{amsmath}
\usepackage{amsthm}
\usepackage{graphicx}
\usepackage{tabularx}
\usepackage{xcolor,xspace,booktabs}

\usepackage{enumitem, multirow}
\usepackage{cite}
\usepackage{booktabs}
\usepackage{placeins}
\usepackage{threeparttable}
\usepackage{amssymb}
\usepackage{pifont}
\usepackage{hyperref,enumerate}
\usepackage[toc,page]{appendix}

\newcommand{\nvper}{\textit{JASS}\xspace}

\newtheorem{theorem}{Theorem}[subsection]
\newtheorem{lemma}[theorem]{Lemma}

\usepackage[normalem]{ulem}

\begin{document}

\title{\textit{JASS}: A Flexible Checkpointing System for NVM-based Systems}

\author{
  Singh, Akshin\\
  \texttt{akshin.singh@cse.iitd.ac.in}
  \and
  Sarangi, Smruti R.\\
  \texttt{srsarangi@cse.iitd.ac.in}
}

\date{}
\maketitle

\thispagestyle{empty}

\begin{abstract}
	NVM-based systems are naturally fit candidates for incorporating periodic checkpointing
	(or snapshotting). This increases the reliability of the system, makes it more immune
	to power failures, and reduces wasted work in especially an HPC setup. The traditional line
	of thinking is to design a system that is conceptually similar to transactional memory,
	where we log updates all the time, and minimize the wasted work or alternatively the MTTR
	(mean time to recovery). Such ``instant recovery'' systems allow the system to recover
	from a point that is quite close to the point of failure. The penalty that we pay is 
	the prohibitive number of additional writes to the NVM.
	
	We propose a paradigmatically different approach in this paper, where we argue that in most
	practical settings such as regular HPC workloads or neural network training, 
	there is no need for such instant recovery. This means that we can afford
	to lose some work, take periodic software-initiated checkpoints and still meet the goals of the application. The key benefit of our scheme is that we reduce write amplification substantially; this extends the life of NVMs by roughly the same factor. We go a step further and design an adaptive system that can minimize the WA given a target checkpoint latency, and show that our control algorithm almost always performs near-optimally. Our scheme reduces the WA by 2.3-96\% as compared to the nearest competing work.
\end{abstract}

\section{Introduction}
NVMs are becoming commonplace in large data centers~\cite{pmdatacenter1,pmdatacenter2}. They have natural advantages as storage media as compared to hard disks, which include faster speed, higher throughput, and fast random access times. They can also be used in a different avatar: many NVM technologies such as Intel Optane~\cite{optane,optanechar} can masquerade as DRAM and increase the amount of usable memory. As compared to DRAM, the main advantage of NVMs is {\em persistence}, which means that data once written does not get destroyed even if there is a power failure. It is this feature that has been used extensively in the literature to design checkpointing  or snapshotting systems that allow an application or the full system to recover from a known point after failure. This is very important in HPC setups where jobs take a long time and any failure can potentially make one lose months of work. Hence, there is a rich body of work that uses NVMs for checkpointing applications and systems. {\bf Note} that we shall use the terms {\em checkpoint(ing)} and {\em snapshot(ting)} interchangeably in this paper.

The underlying conceptual thread in highly cited prior work in this area \cite{TSOPER,NVOverlay,PiCL,THYNVM,BSP,DHTM,WSP,LAD,REDU,WrAP,HOOP,PROTEUS} is highly cited from similar work that was done in transactional memory (TM). Even though not stated explicitly, the idea is that the systems are executing important computations all the time, and as a result, continuous logging and {\em instant recovery} are required. A system is said to be instantly recoverable if the point of re-execution (after a failure) is very close to its point of failure. Such assumptions are undoubtedly necessary for applications such as the stock market or e-commerce platforms, where there is a legal requirement to log everything and we have millisecond-level transactions~\cite{HFTtrade1}. However, there are also a very large number of applications that do not have such requirements. Consider, the case of NN training, weather simulation, molecular simulation, and HPC jobs of a similar nature.  Let us say that after failure recovery, we lose 2 seconds of work; users will fail to notice it. An analysis of the Google cloud trace by Ahmed et al.~\cite{ahmed2020characterizing} indicates that for most servers the MTTF (mean time to failure) is roughly 10 days and the MTTR (mean time to recovery) is high -- roughly 5:44 hours. Other studies have indicated a failure rate of 215 FITs(1 FIT = 1 failure/billion hours)~\cite{FITPAPER,CONCRETEFIT}. In such a scenario, saving a few milliseconds of checkpointing overhead is not of any consequence. Far bigger practical concerns dominate such as extending the life of the system, and ensuring a lower cost of ownership.  

In this context, the traditional continuous checkpointing model \cite{PiCL,THYNVM,NVOverlay} is not suitable because it increases the write amplification (additional writes to the NVM to make the system instantly recoverable) substantially, which was not a concern when we were writing to the caches and DRAM in the TM era. Caches and main memory were assumed to have infinite endurance (for all practical purposes). For NVM-based systems, write amplification is very important. Given that NVM devices have a limited shelf life, they should only service genuine writes generated by the application. We observe in this paper that the price of instant recovery is roughly a 3$\times$ increase in the write amplification.

We are thus proposing a paradigmatic shift to the NVM checkpointing problem. Our vision is more in line with the fundamental requirements of HPC and ML applications, where everything need not be logged and we have idempotency (no harm in re-executing the same computations). We instead argue that checkpointing can be done periodically using a timer interrupt or a software signal, or for a distributed application it can be initiated by an external node. It should be an infrequent operation: have a high epoch size ($ES$). This automatically reduces the $WA$ and further leads to a tradeoff between the checkpoint latency ($CL$) and the write amplification ($WA$). The {\em checkpoint latency} is defined as the time it takes to complete the checkpoint after initiating the process. The $CL$ and $WA$ are inversely related: one increases as the other decreases. We would like to set a bound on $CL$ because it determines the responsiveness of the system. Many a time, we would like to take a checkpoint before an I/O operation or system call;  there may be a requirement that the system needs to be quiescent till the checkpointing finishes. The problem that we thus solve is {\bf given a $CL$ and $ES$, create an architecture to minimize the $WA$}.

Our proposal \nvper embodies the basic realization that the conflict between $CL$ and $WA$ is like an insurance policy. The higher is the insured amount, the higher is the premium, which one must pay continuously throughout the lifetime of the policy. It is up to the policyholder to figure out whether the expenditure on the premium is justified given the insured amount. The same is the case for us, where the $WA$ is something that we incur throughout the life of the application. We shall see in Section~\ref{sec:results} that to achieve a reduction of even 1 ms in the $CL$, we will have to pay a big price in terms of the $WA$ (2 $\times$) because we always need to be ready to finish the checkpoint within the bound. Our first contribution is a system that allows us to operate anywhere in the region of feasible $\langle CL, WA \rangle$ for a workload and given value of $ES$ -- given a $CL$ value, we make a best effort to meet it and also minimize the $WA$. 

Many more innovations are possible because of this design philosophy. Existing NVM-based checkpointing schemes~\cite{THYNVM,NVOverlay,PiCL} checkpoint at the granularity of epochs (like us), however, while doing so, they introduce a very complex system that requires us to tag cache and sometimes DRAM data with multiple bits, persist when there is a coherence write or eviction, maintain the state of many epochs at the same time and have elaborate data coalescing algorithms. We design a scheme that is far simpler, where the crux of our algorithm is a coherent cache flushing scheme using a variant of the classical Chandy-Lamport algorithm~\cite{nancy} for distributed systems. We need to maintain only a single bit per cache line, require no changes to be made to the DRAM; we propose a simple closed-loop control algorithm to ensure that the $CL$ target is met while minimizing the $WA$.
We compare \nvper with the nearest competing work, NVOverlay \cite{NVOverlay}, and show that for its constraints, our $WA$ is 3$\times$ lower. Given a $CL$ target, we never overshoot (or undershoot) by more than 5\%.

\textbf{Main ideas in our proposal: \nvper}
\begin{itemize}
\item A novel cache flushing scheme that flushes all in-flight pre-snapshot messages from the NoC and pre-snapshot data from the
caches using a variant of the Chandy-Lamport algorithm.
\item A locality predictor that decides whether to keep a page in DRAM or not based on some tunable hyperparameters.
\item A method to dynamically tune those hyperparameters  such that the target CL is achieved while minimizing the WA.
\end{itemize}

Section~\ref{sec:back} describes the relevant background, Section~\ref{sec:motiv} presents the motivation, Sections
\ref{sec:overall-design} and \ref{sec:detailed-design} elaborate on the design; we present our results in 
Section~\ref{sec:results}, described related work in Section~\ref{sec:rel} and finally conclude in Section~\ref{sec:conc}.

%\begin{enumerate}
%\item NVOverlay presumes tags in main memory (not realistic)
%\item Too many versions. We don't need that many. Managing them is complex.
%\item Data coalescing not shown, this is an issue.
%\item Write amplification is too much (cannot be tailored based on the requirements).
%\item No global synchronization between epochs on different cores. 
%\end{enumerate}

%Our main contributions are:

%\begin{enumerate}
%\item A globally synchronized checkpoint without stalling the system.
%\item Single bit added to each cacheline's state to track lines.
%\item No monitoring coherence traffic.
%\item Switchable mode of operation that lets users make a tradeoff between write-amplification and latency. 
%\item On-demand snapshots.  
%\end{enumerate}

%FIXED: All references to ES, CL, and WA are in $$

\section{Background}
\label{sec:back}

The nonvolatile nature of byte-addressable Nonvolatile Memory (NVMs) has inspired different approaches to use it. Most software approaches perform some form of logging or transactions, while most hardware approaches define and implement the order of stores to memory.
In order to understand these approaches, we must look at the characteristics of NVMs first. 

\subsection{NVM Cell Characteristics}

All popular types of NVM cells (ReRAM, FeRAM, STT-RAM, PCM, NAND flash) have a smaller write endurance (roughly 1000X lower) when compared to DRAM~\cite{hardwaresurvey}. When the number of writes to a cell exceeds the endurance threshold, the cell loses its ability to retain data without applied power ({\em persistence}). Since the major benefits of using an NVM device is {\em persistence}, repeated writes to same cells are extremely problematic. As a result, device side \textbf{wear-levelling} is implemented to ensure that every cell sees roughly the same number of writes.

\subsection{Wasted Writes}
Although NVM devices placed on the memory bus are byte-addressable, they have an internal hardware buffer (for reading and writing) that caches data at the granularity of blocks~\cite{optanechar}. Writes to this block are coalesced by this buffer. A block switch causes a write-back of a modified block to a cell at a different location. This relocation of blocks is part of a device-side wear-leveling technique. We formally define {\em write amplification} as a
ratio: number of block-level writes made by the target system divided by the number of byte-level writes {\bf intended} by the target system. 

\subsection{Place in the Memory Hierarchy}

Since NVMs are a new memory class, their position in the hierarchy is critical. Figure \ref{fig:placement} shows the different options. Most research papers place DRAM and NVM at the same level~\cite{survey}, i.e., horizontal integration. This provides greater flexibility, as applications can divide their address space into persistent and non-persistent domains. The resulting problem of \textit{data placement} is handled by software \cite{espresso,autopersist}.
A system can opt for {\em replacement integration}, replacing DRAM with NVM. Although cost-effective, such a system suffers from long read and write delays to main memory. In this paper, we opt to design a system with the more popular horizontal integration. 

\begin{figure}[htbp]
     \centering     
     \includegraphics[scale=0.35]{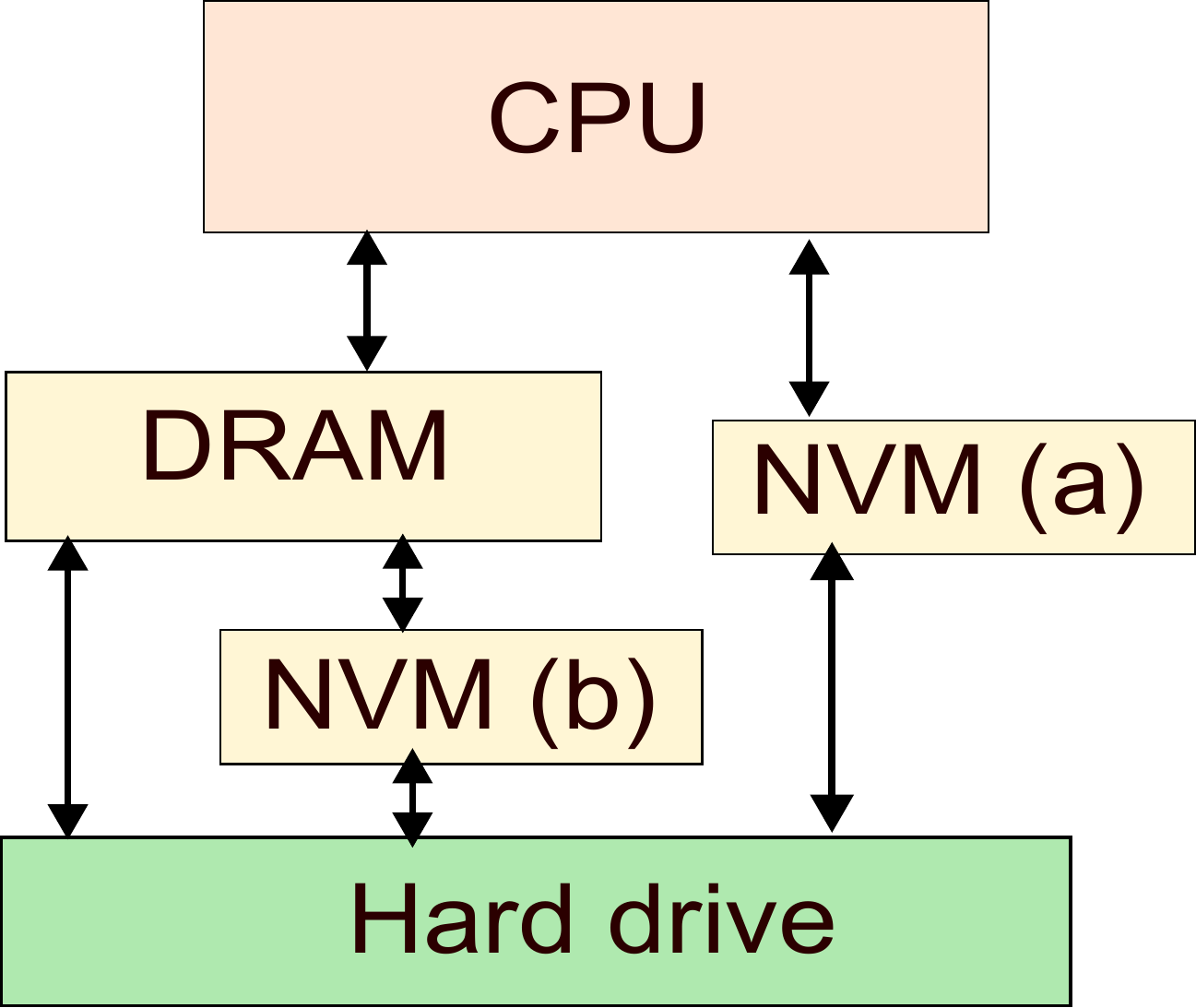}
     \caption{Placement of NVM in the memory hierarchy. (a) Horizontal integration puts NVM and DRAM on the same level. (b) Vertical Integration places NVM below DRAM.}
     \label{fig:placement}
\end{figure}

\subsection{Persistency Models and Architectures}
In 2014, Pelly et al. \cite{Pelly} defined three different persistency models: strict (writes follow a well defined persist order), epoch (persistence is guaranteed only at epoch boundaries) and strand persistency (threads can be broken into strands and two strands are independent). These models are helpful to software that accesses an NVM device directly since these models provide certain guarantees on the allowed order of persists (that are visible post-crash), conceptually similar to the visibility of writes by other threads in a memory consistency model~\cite{advarch}.  

\nvper is a full-system \textit{user-transparent} snapshotting system that persists at epoch boundaries, just like its predecessors \cite{TSOPER,ATOM,NVOverlay,PiCL,THYNVM}. The system can be seamlessly started from its last epoch boundary. The computations done after the last epoch boundary need to be redone, here idempotency is assumed. Most HPC workloads have this property.

The key idea in any epoch-based checkpointing system is to tag all the writes that belong to a periodically-increasing epoch with an epoch id. This epoch id is stored in cache lines, and in the case of NVOverlay~\cite{NVOverlay}, in the DRAM rows as well (in place of the ECC bits), which is very impractical in our view. This means that DRAM errors cannot be corrected and this also violates the semantics of DRAM usage, which has resisted change for decades. Now, when a read from a remote core accesses data, and the epoch id of the block is more than the epoch id of the remote core, the remote core gets promoted to the new epoch id. Data is typically persisted to the NVM when either there is a coherence write or when data from older epochs is evicted from the memory hierarchy. The issue is that we need to maintain a lot of epoch versions at the same time and before writing the data for an epoch, there may be a need to coalesce the data with the modifications made in all previous epochs. All of these issues lead to a large number of simultaneously-alive epochs and complex state management. Our aim is to simplify this and just maintain two epochs at a time. We rely on the classical Chandy Lamport algorithm~\cite{nancy}.

\subsection{Chandy-Lamport Snapshot}
\nvper collects a snapshot of a running system using a (modified) Chandy-Lamport distributed snapshotting algorithm as shown in Algorithm~\ref{clalgo}. The Chandy-Lamport snapshotting algorithm is initiated over the NoC by any node. It could also be initiated upon a message received from a remote machine (in the case of distributed coordinated checkpointing). Links in the network are assumed to be FIFO, which is the case for an NoC.  

\begin{algorithm}
  \footnotesize
  \caption{Chandy-Lamport algorithm}
  \label{clalgo}

    \begin {algorithmic}[1]
    \If{Snapshot token received for the first time}
    \State Take local snapshot.    
    \State Propagate token to all neighbors.
    \State Send an acknowledgement to the sender. 
    \ElsIf{Snapshot token received more than once}
    \State Send an acknowledgement to the sender.
    \ElsIf{Token received and message arrives on a non-token/non-ack receiving link}
    \State Record message.
    \EndIf
  \end {algorithmic}
  
\end{algorithm}

There are two kinds of messages: pre-snapshot (initiated before the snapshot token was received by the sender) and post-snapshot
(initiated after the sender received the snapshot token). 
The main idea is to propagate token messages on every link. If a node has received a token, it acknowledges the 
receiving of the token. Upon the first arrival of a token at a node, a local snapshot is taken at that node. 
This snapshot records the local state. 
Since the arrival of a token signifies that the sender has taken a snapshot, any subsequent messages on that link
are post-snapshot messages and are not recorded.

\section{Motivation}
\label{sec:motiv}

In this section, we motivate our design decisions. We start by looking at the issues that bedevil a software approach to checkpointing. We then look at two different components of a checkpointing system not discussed in previous works: latency and correctness. We conclude this section by arguing in favor of bigger epochs. 

\subsection{Issues with Software-based Checkpointing}
Software-based checkpointing solutions are created in one of two ways. One option is to change the algorithm of programs and modify them to be fault tolerant, by manually taking checkpoints (could be compiler-driven as well). This will ensure forward progress but it has high overheads~\cite{softwareoverheads1}. Moreover, it inflicts a performance penalty on the system since software-implemented \textit{persistence} needs to be realized.   
The other option is to take a snapshot without modifying the original program. This is achieved by pausing the execution, collecting a copy of the registers and memory, and finally persisting this data on to stable storage. The degradation in performance with this approach is attributed to the interruption of running programs. Now, since failures are expected to be rare, we would ideally like all checkpointing activity to be off the critical path. 

Sadly, many a time before a checkpoint completes, we cannot start some activity like a system call, hence having hardware
support to reduce checkpoint latency is crucial because the checkpointing operation itself may inevitably be on the critical path.

\subsection{Latency of a Checkpoint}

\begin{figure}[htbp]
     \centering     
     \includegraphics[width = 0.8\columnwidth]{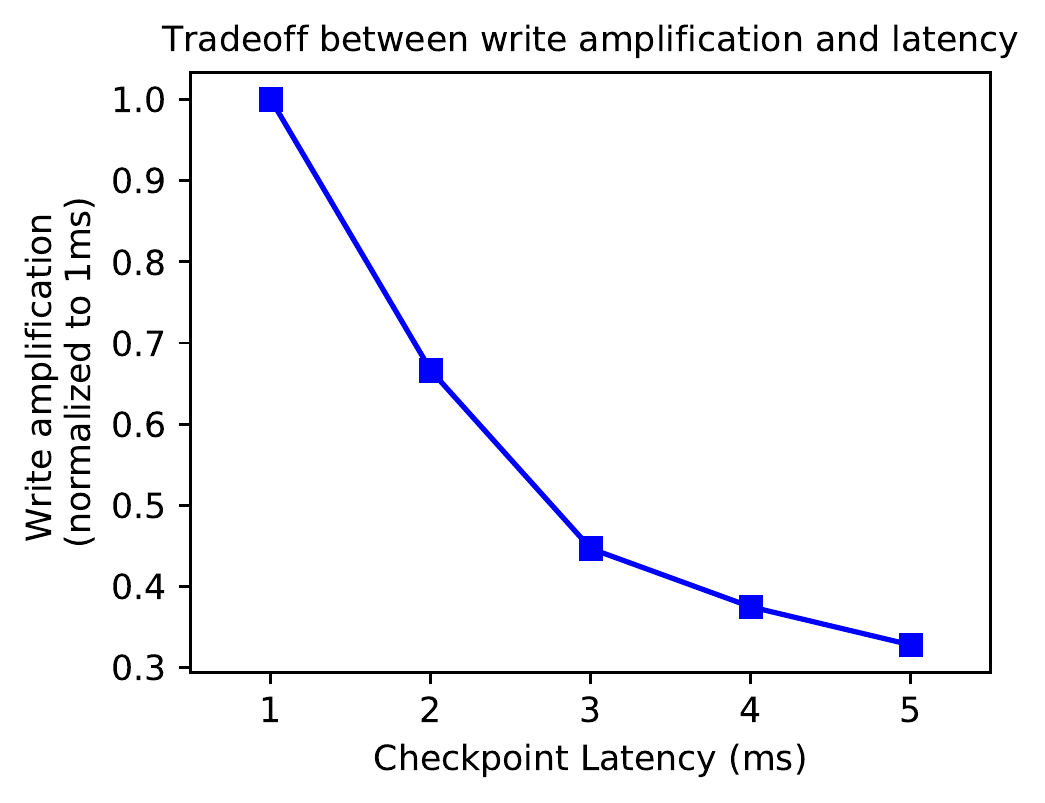}     
     \caption{Tradeoff between write amplification and latency (representative figure for the {\em facesim} benchmark)}
     \label{fig:tradeoff}
\end{figure}

We characterize that the latency of a checkpoint, i.e., the time taken by a checkpoint to complete after it is initiated, as one of the most crucial characteristics of the a system. Internal and external events that depend on the availability of a checkpointed system, will perceive the responsiveness of a system to be a function of the checkpoint latency ($CL$). We aim to design a system that meets a target latency provided by the user and then try to minimize write amplification ($WA$). Figure \ref{fig:tradeoff} illustrates the tradeoff achieved between $CL$ and $WA$ for the {\em facesim} benchmark (more details in Section~\ref{sec:results}) with \nvper. As expected, with a lower $CL$, the system observes fewer writes to the NVM.  Note how diminishing returns quickly set in for higher values of $CL$. It is up to the user to provide a $CL$ target, which is
a function of the responsiveness of the system and the epoch size ($ES$).

\subsection{Issues with Existing Approaches}
Existing approaches maintain the state of many epochs ({\em versions}) at a time, which complicates the caches. NVOverlay
goes one step further and proposes to re-purpose the ECC bits in DRAM rows to store version information, which would by and large
not be acceptable to the DRAM industry. Second, to ensure a small $ES$, writes are frequently sent to the NVM leading to higher write
amplification and given that many versions need to be reconciled across the cache hierarchy, data coalescing becomes quite complicated.
There are more subtle issues: some works~\cite{PMEMSPEC} assume
separate persist (NVM) and store (DRAM) paths in the system, which is not practical. We use the same path, i.e., persist and store via the NoC
(similar to~\cite{WSP,BSP}).  We observed that correctly flushing the caches and in-flight messages in the NoC is a non-trivial
problem, which has not been adequately dealt with in prior work~\cite{NVOverlay,THYNVM,PiCL}. We propose a rigorous theoretical
mechanism that can efficiently do so.

\section{\nvper : Overall Design}
\label{sec:overall-design}

\nvper is a full-system snapshotting system. The overview of the system is shown in Figure~\ref{overall}. Given that we use long epochs, there is no need to maintain the state of multiple epochs at a time. We maintain only 1-bit state using a single bit: pre-snapshot and post-snapshot. At the end of the next epoch, the sense {\em reverses} (post-snapshot in the previous epoch becomes pre-snapshot in the next epoch). Every cache line is tagged with this bit; however, main memory (DRAM) is left unmodified.
Both epochs have separate DRAM page (simply called a page henceforth) tables that refer to the modified working sets of the epochs.
A DRAM scrubbing scheme persists the modified pages for the epoch that is being persisted. {\bf Note} that a {\em page} is defined as
a 256-byte region of contiguous memory (not the 4 KB virtual memory page, which we refer to as a VM-page.)

\begin{figure}[!htb]
  \centering
     \includegraphics[width=0.45\textwidth]{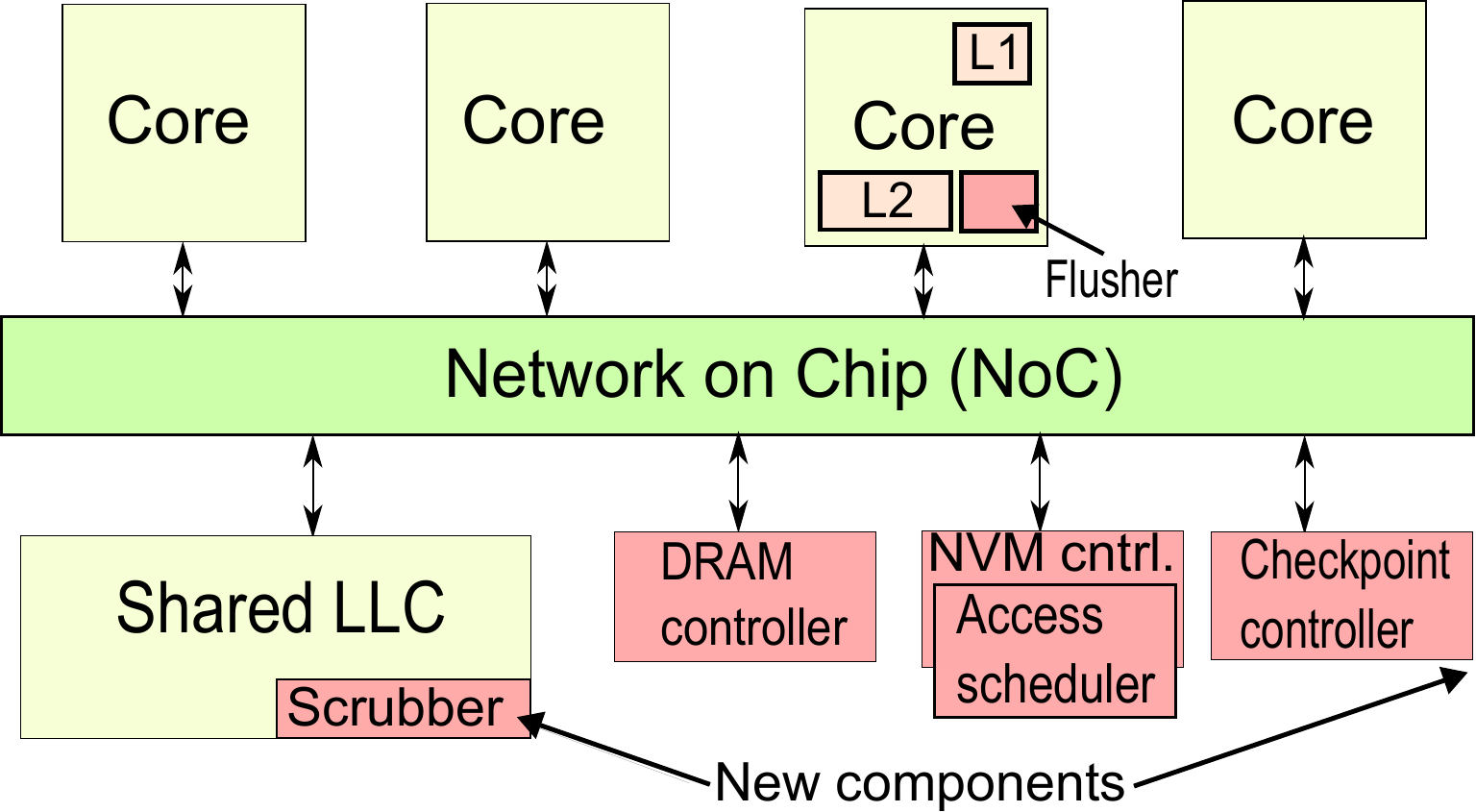}
   \caption{Design of the overall system. The red (darkly shaded) regions are the new components in \nvper.}
  \label{overall}
\end{figure}

\subsection{Regular Operation}
While not taking a checkpoint the operation of the system carries on as usual. Given an $\langle ES, CL \rangle$ 
target, we are not allowed to exceed the $CL$ bound and we need to do our best to minimize the write amplification, $WA$. 
Hence, there is always a bound on how much of data we can keep in a volatile state (in the caches and memory). We
thus create a temporal locality predictor that continually moves data that is not expected to be accessed in the near future
to the NVM such that when there is a request for a checkpoint, we need to do a bounded amount of work. This does 
increase $WA$ but if the locality predictor is effective, the increase is minimal.

\subsection{Checkpoint Initiation}
\nvper has two ways to initiate a checkpoint. A clock-driven method takes a periodic checkpoint of the system. If the system decides that, in the rare event of a failure, it is okay with losing a maximum of $\eta$ minutes of work, we set the period between checkpoints to $\eta$ minutes. This is the
epoch size, $ES$. The system sets this value, and we expose this using a privileged instruction.
This mode can be turned off if periodic checkpoints are not desirable.
The other method is to initiate a checkpoint before an {\em event of interest} like a system call or I/O instruction. A machine
can also initiate a checkpoint in a large distributed system when it gets a message from a remote node.
 
\subsection{Process of Taking a Checkpoint}
We can divide the process of taking a checkpoint into three distinct phases: \textcircled{1} Flush the pre-snapshot messages in the NoC and 
the lines in the high-level caches to the LLC, \textcircled{2} scrub the caches and write the changes to the NVM, \textcircled{3} scrub the main
memory and move all the modified pages (since the last epoch boundary) to the NVM. 

\subsection{Incremental Checkpoints and Recovery}
\nvper takes incremental checkpoints. The NVM device has a page table in a known location that is 
used during recovery. We use atomic logging (similar to~\cite{NVOverlay}) to ensure recovery of the page table data structure. 
Post-crash recovery is simple but time-consuming. First, we read the page table and move pages to DRAM one by one. 
When this is done, we restore the register state in all the cores. The system can then resume.

\section{\nvper : Detailed Design}
\label{sec:detailed-design}

As described in the previous section, a checkpoint can be initiated for a variety of reasons:
period timer interrupt, before a system call or I/O operation, or as a part of a coordinated activity
in a distributed system. Regardless of the reason, the flow of actions is the same. A checkpoint controller
starts the process and initiates the Chandy Lamport algorithm by sending a checkpoint message (snapshot token) to its
neighboring nodes on the NoC.

\subsection{The Process of Snapshotting}

Upon receiving a snapshot token at a router on the NoC, the router checks if it has received one before. If it has not received a token, the connected elements are notified to take a snapshot, and the router propagates the token to all neighbours. 
Every core has a single bit initialized to zero. This bit is attached to every message and cache line in the system. Let us call this bit the \textit{snapshot} bit. Consider the case when we are taking the first snapshot.

If the snapshot bit is $0$, we know that the message/cache line associated with the bit is pre-snapshot. Once an element is made aware of an ongoing snapshot, all further message generation and cache modifications happen by appending a $1$ bit instead of a $0$ bit, thereby distinguishing between pre-snapshot and post-snapshot data/messages.

After snapshot completion, we can change the sense of the bits: now $1$ stands for a pre-snapshot state and $0$ stands
for a post-snapshot state.

\subsubsection{Snapshotting Cores and Private Caches}
Snapshotting a core is simple. We first flush the pipeline and then store the architectural register state in a known location in the NVM. The core is thus assumed to be checkpointed. All subsequent memory writes are deemed to be post-snapshot writes. 
Snapshotting a cache also entails changing the state of future evictions to post-snapshot writes. We assume that both the cores
and caches contain a snapshot bit, whose sense changes every epoch.

\subsection{Flushing Messages on the NoC}
Correctly flushing pre-snapshot messages from the NoC is complicated and even though it looks simple, it is not so. 
This is something that has not gotten its due share of importance in prior work. Many a time architects assume that a solution is simple and something at design time can be worked out; however, this is one such instance where this assumption is not true. There are a lot of hidden subtleties and doing this correctly is quite difficult.
It is possible that a pre-snapshot message gets stuck for a long time in a router. We must make sure 
that this message reaches its destination before we start scrubbing data in the lower-level caches. To achieve this, we must perform some form of {\em termination detection} that makes sure that there are no messages in the NoC with a pre-snapshot tag.

\begin{algorithm}
  \footnotesize
  \caption{Flushing Algorithm}
  \label{flushalgo}
  
  \begin {algorithmic}[1]
    \State Initially, $\forall i, state(R_i)  = NTR$ \algorithmiccomment{No token received}
    \State Router $R_i$ receives message $M$
    \If {$state(R_i) = NTR \land type(M) = TOK$ } \algorithmiccomment{First token received}
    \State $state(R_i) = TR$

    \State $xcount = 0$
    \For {$M_j \in buffer(R_i)$}
    \State $mark(M_j)$
    \State $xcount \gets xcount + 1$ \algorithmiccomment{Mark all in-flight messages}
    \EndFor

    \State {\bf Send} $xcount$ to the checkpoint controller,                 
    $pcount \gets 0$

    \For{$node \in neighbors(R_i) \cup tile\_elements(R_i)$}
    \State {\bf Send} $TOK$ to $node$ \algorithmiccomment{Propagate to all neighbors}
    \EndFor
    
    \ElsIf {$state(R_i) = TR$} \algorithmiccomment{Already received}
    \If{$dest(M) = i \land type(M) \ne TOK$} \algorithmiccomment{Leaving the NoC}
    \State $pcount \gets pcount + 1$    \algorithmiccomment{$pcount$ is the msg processed count}
    \EndIf    
    \EndIf

    \State Periodically:
    \State \hspace{\algorithmicindent} Send $pcount$ to the checkpoint controller

    \If {$\forall i, R_i$ has taken a checkpoint $\land  \sum pcount_i  = \sum xcount_i$ } 
    \State \textbf{Flushing complete}
    \EndIf
  \end {algorithmic}
\end{algorithm}

Algorithm \ref{flushalgo} describes the operation of routers in the presence of an ongoing snapshot. $R_i$ describes the $i^{th}$ router. {\em buffer} represents the internal buffers of a router. Each router can be in one of two states: token recieved ($TR$) or no token recieved ($NTR$). When the router is in state $NTR$, and the incoming message is not a token, normal operations continue. Algorithm~\ref{flushalgo} augments the classical Chandy-Lamport algorithm~\cite{nancy} with a few additional actions (see Section~\ref{sec:back}). To start with, note that the checkpoint token gets the {\bf highest priority} in the routers, which preserves
the abstraction of FIFO channels with respect to checkpoint correctness.

The main idea is to keep a count of the all pre-snapshot messages in the system and ensure that all of them reach their destinations before we begin the next phase. We achieve this in a distributed way, having each router count messages internally (expressed in the $xcount$ variable), and later sharing this information with the checkpoint controller, which stores the value. Periodically, all the destination routers send the count representing the total number of messages processed ($pcount$) to the checkpoint controller. When $\sum pcount_i$ is the same as $\sum xcount_i$ and all the routers successfully take their checkpoint, we can conclude that all pre-snapshot messages in the NoC have reached their destination \footnote{There are subtle issues with in-flight messages and coherence directories. We discuss them in Appendix \ref{appendix.proof}} . 

\begin{theorem}
  Algorithm \ref{flushalgo} will correctly flush all pre-snapshot messages from the system. 
\end{theorem}

\begin{proof}
  Refer to the results in  appendix \ref{appendix.proof}. 
\end{proof}
\subsection{Cache Operations }

After we take a snapshot at a cache, normal operations continue. The pipeline or upper levels will keep sending requests, albeit marked post-snapshot, and the caches will continue to service them. However, we must reconsider some cache operations that may depend on the snapshot bit. The updated cache operations are summarized in Algorithm \ref{cacheoperations}.

\begin{algorithm}
  \footnotesize
  \caption{Cache operations}
  \label{cacheoperations}
  
  \begin {algorithmic}[1] 

    \State Message received $M$ for a cache line $l$.
    \If{$type(M) = READ$}
      Read $l$
    \ElsIf{$type(M) = WRITE$}
    \If{$l$ is marked \texttt{pre-snapshot} and $l$ is modified}
          \State {\bf Send} $l$ to the next level with the \texttt{snapshot} bit .
       \EndIf
    \State Write to $l$.

    \ElsIf{$type(M) = GETS$}  \algorithmiccomment{Get a copy of the line in shared state}
    \State Forward $l$ to a requesting cache
    \ElsIf{$type(M) = GETX$} \algorithmiccomment{Get a copy of the line in exclusive state}
    \If{$l$ is marked \texttt{pre-snapshot} and $l$ is modified}
    \State Send $l$ to the next level with \texttt{snapshot} bit
    \EndIf
    \State Forward $l$ to the requesting cache; Evict $l$
    \ElsIf{$type(M) = INVALIDATE/EVICT$}
    \If{$l$ is marked \texttt{pre-snapshot}}
    \State Send $l$ to the next level with the {\texttt snapshot} bit
    \EndIf
    \State Evict $l$
    \EndIf
        
  \end {algorithmic}

\end{algorithm}

\textbf{Read}: Reading from a line does not interfere with our snapshot. Since future writes on the same core will not be a part of the checkpoint, locally servicing a read even from a pre-snapshot marked line is correct.  

\textbf{Write}: Locally servicing a write to a line may overwrite the data that is marked pre-snapshot. As a result, we must evict the line to the lower level before we effect the write. If the lower level receives this write but has a line marked as pre-snapshot, we overwrite that value because values at higher levels are more recent. In this way, the cache hierarchy coalesces writes internally.

\begin{figure}[!htb]
  \centering
  \includegraphics[width=0.8\columnwidth]{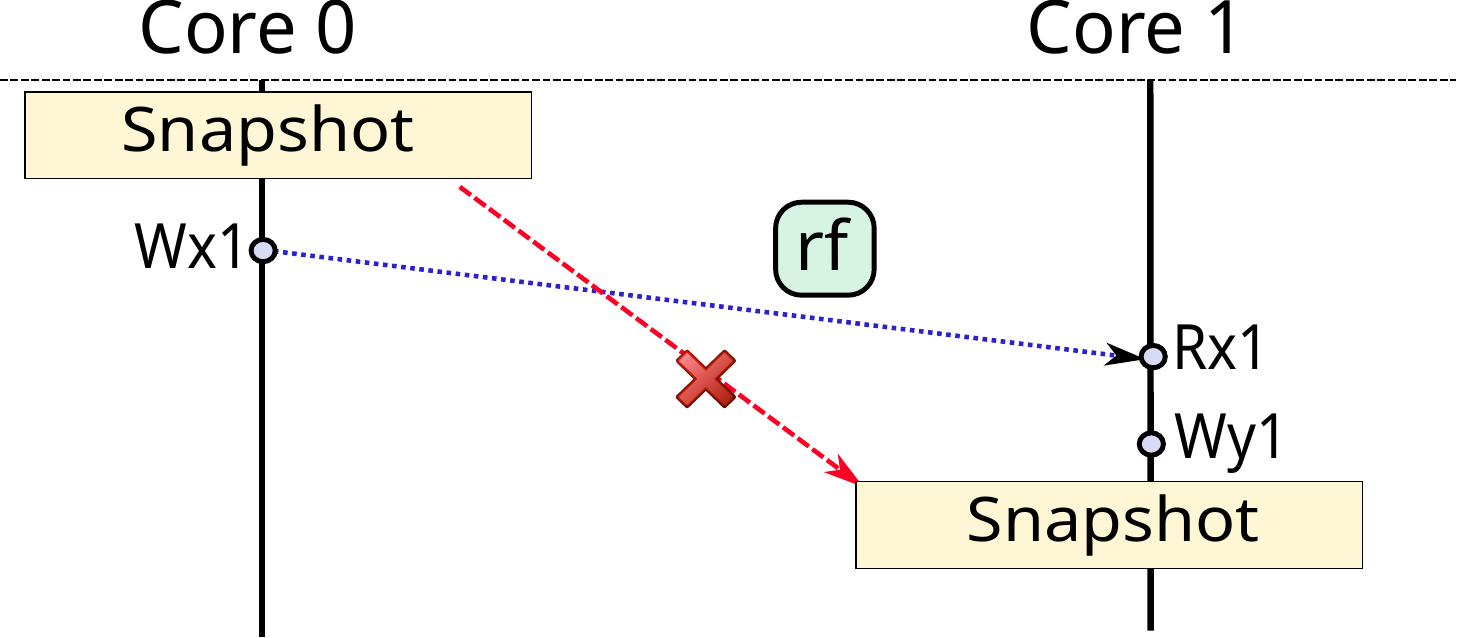}   
  \caption{An $rf$-edge (read-from) cannot cross a snapshot token}
  \label{GETS}
\end{figure}

\textbf{GETS}: A GETS (get a copy of the block in the shared state) request does not lead to a write to the line. However, one might argue that a GETS can cause inconsistency in snapshotted data. Figure~\ref{GETS} illustrates the problem. Consider a write to $x$ ($Wx1$) by core $0$. We then have a read to $x$ by core $1$. Core $1$ reads $x$ as $1$, thereby creating an $rf$ ({\em read-from}) edge. Now, at core $1$, a  write to $y$ executes. We will have an inconsistent snapshot if $Wx1$ is not included in the snapshot while $Wy1$ is included. 

The above situation implies that the writer (\textbf{Wx1}) is post-snapshot, since it is not included in the snapshot while a causally related write (\textbf{Wy1}) is included in the snapshot. However, this situation can never arise. Since the token has the highest priority in our network (discussed above), the reading cache will receive the token from the sister cache before a response from that sister cache. As shown in Figure \ref{GETS}, in our system, a token cannot cross an \textbf{rf}-edge in a hypothetical dependence graph. Consequenlty, \textbf{GETS} messages can be serviced without eviction.

\textbf{GETX}: A GETX (get a copy in the exclusive state) message will write to the line. As a result, we must evict this line to the lower level before forwarding it to the requesting cache.

\textbf{EVICT/INVALIDATE}: We send pre-snapshot information to the next level along with the line. 

\subsection{Cache Scrubbing}

\begin{figure}[!htb]
  \centering
  \captionsetup{justification=centering}
  \includegraphics[width=0.8\columnwidth]{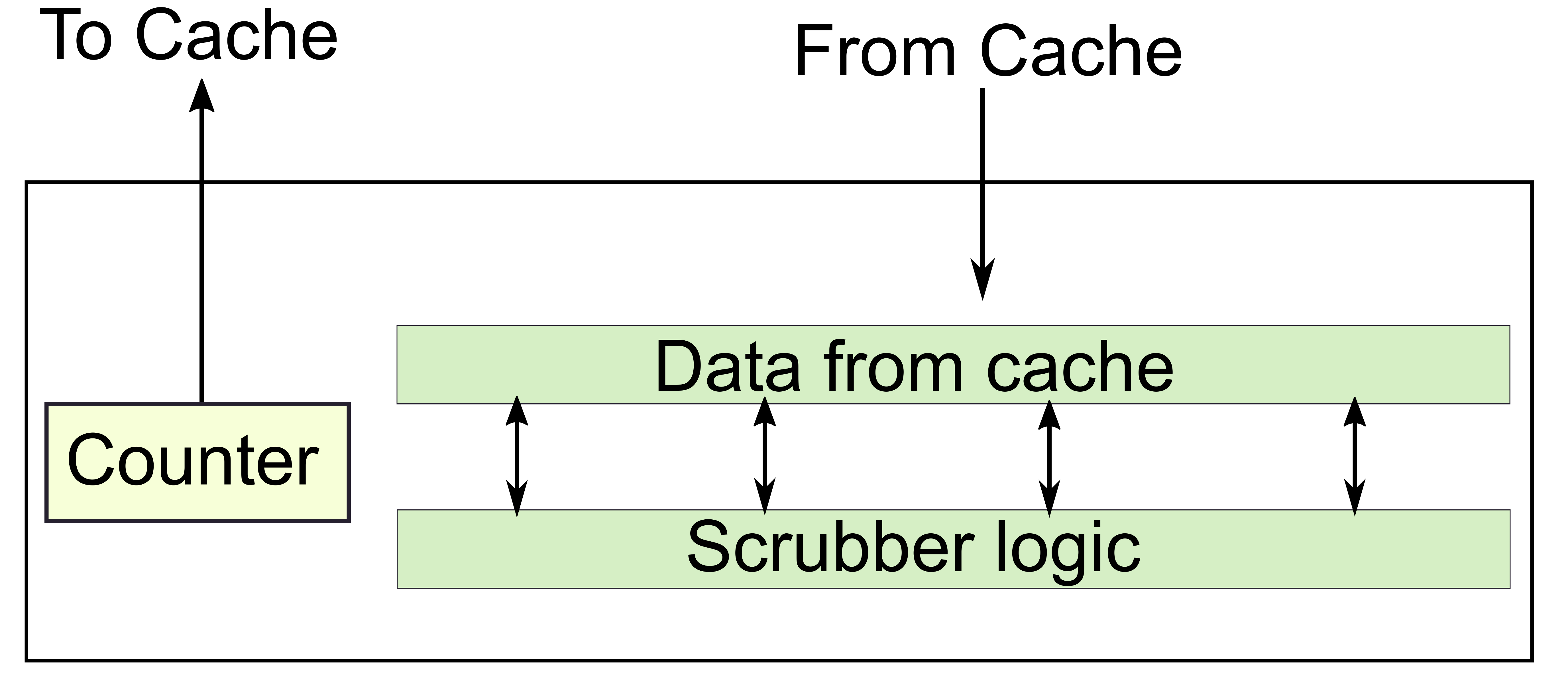}   
  \caption{The implementation of the scrubber}
  \label{SCRUBBERIMAGE}
\end{figure}

Figure \ref{SCRUBBERIMAGE} shows the implementation of our scrubber in a cache. The idea is to read entire rows and send pre-snapshot data to lower levels. Scrubbing can be made more efficient by delaying the time between consecutive scrubs. However, the checkpoint's latency depends on our scrubber's frequency since a long-lived pre-snapshot line can be present in the last row of our cache (one that the scrubber touches the end). We, therefore, present a tunable parameter called \textit{scrubbing-step}. We perform a scrubbing cycle of the cache at each \textit{scrubbing-step}. The lower its value, the more the scrubber activates.

To avoid a scenario where an evicted pre-snapshot line reaches a lower level after the lower level has scrubbed the set for that line, we serialize the scrubbing at different levels. First, we scrub  the L2 cache after flushing finishes (in our setup). When this is over, a broadcast from the checkpoint controller tells all tiled shared caches to scrub. Serializing scrubbing in this way does not affect the latency of our checkpoints since the DRAM turns out to be the bottleneck.

\subsection{NVM Access Scheduling}

Since \textit{write amplification} is defined by the number of wasted writes in an NVM device due to its block-based wear-levelling techniques, we coalesce data in the NVM controller. We propose the implementation of a 64 cache line (16 DRAM pages) access scheduler for our NVM, not unlike the ones used in modern GPUs~\cite{GPUAS} (our design is shown in Figure \ref{accessscheduler}). We use an access scheduler since the goal of this structure is to minimize the number of row switches. GPUs use such a structure to increase the effective bandwidth of the main memory. We use it to reduce the number of row switches and the number of wasted writes. 

\begin{figure}[htbp]
  \centering
  \captionsetup{justification=centering}
  \includegraphics[width=0.99\columnwidth, height=1.8in]{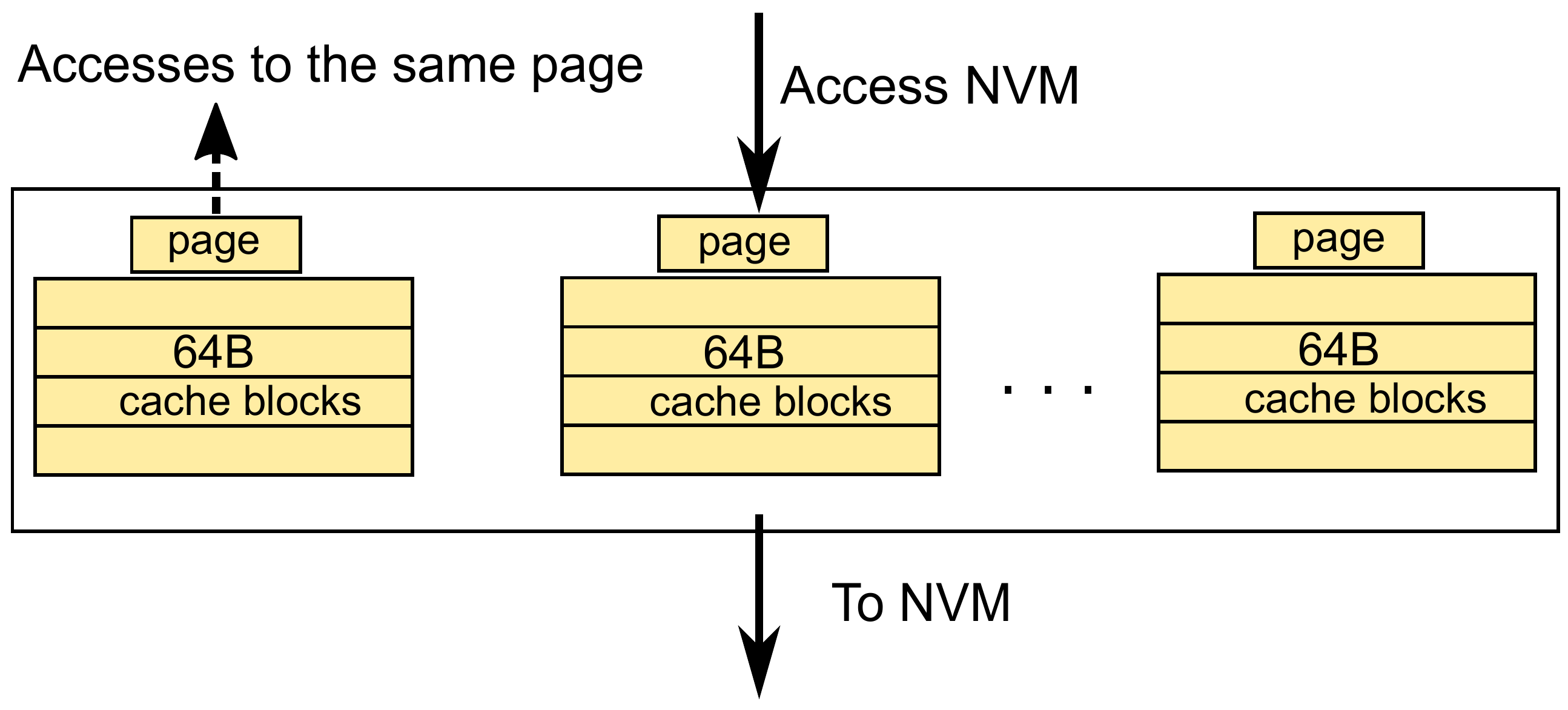}   
  \caption{Implementation of the access scheduler.}
  \label{accessscheduler}
\end{figure}

The access scheduler maintains the state for 16 pages. For each page, we store its 4 constituent cache blocks (contiguous physical addresses). This structure can coalesce all the writes to the same page. Given that our page is a strict subset of a 4-KB VM page, our assumptions or contiguity do not pose a problem.

\subsection{DRAM Scrubbing, Prediction and Coalescing}
In this section, we discuss DRAM scrubbing, which is required to persist pages in the DRAM that are a part of the epoch that needs to be persisted. This need arises because we do not persist all writes to the NVM by default. This is a direct consequence of our decision to have large epoch sizes and bound the $CL$. The research questions emerge.

\begin{itemize}
\item RQ1: How and when to scrub the DRAM?
\item RQ2: How to enforce an upper-bound on the number of unpersisted pages in the DRAM?
\item RQ3: How to merge the DRAM and cache data before writing to the NVM?   
\end{itemize}

\subsubsection{RQ1: DRAM Scrubbing}
If we were to scrub the DRAM similarly to the caches, we would need to walk and read all of the DRAM for every epoch, which is not only wasteful but also unnecessary, and it makes maintaining a bound on the $CL$ difficult. Any epoch will populate a fraction of the DRAM that is a part of the modified working set of that epoch. Due to the incremental nature of our snapshot, the pages in the modified working set need to be persisted. 

We must, therefore, track the working set of each epoch. We use a per-epoch DRAM page table that only tracks the modified pages; {\bf note that in the DRAM no page address re-mapping/translation is done}. This page table is similar to page tables used by the OS (albeit conceptually). However, unlike OS page tables, this page table is for the entire system that covers all used physical addresses.

\begin{figure}[htbp] 
  \centering
  \captionsetup{justification=centering}
  \includegraphics[width=0.99\columnwidth]{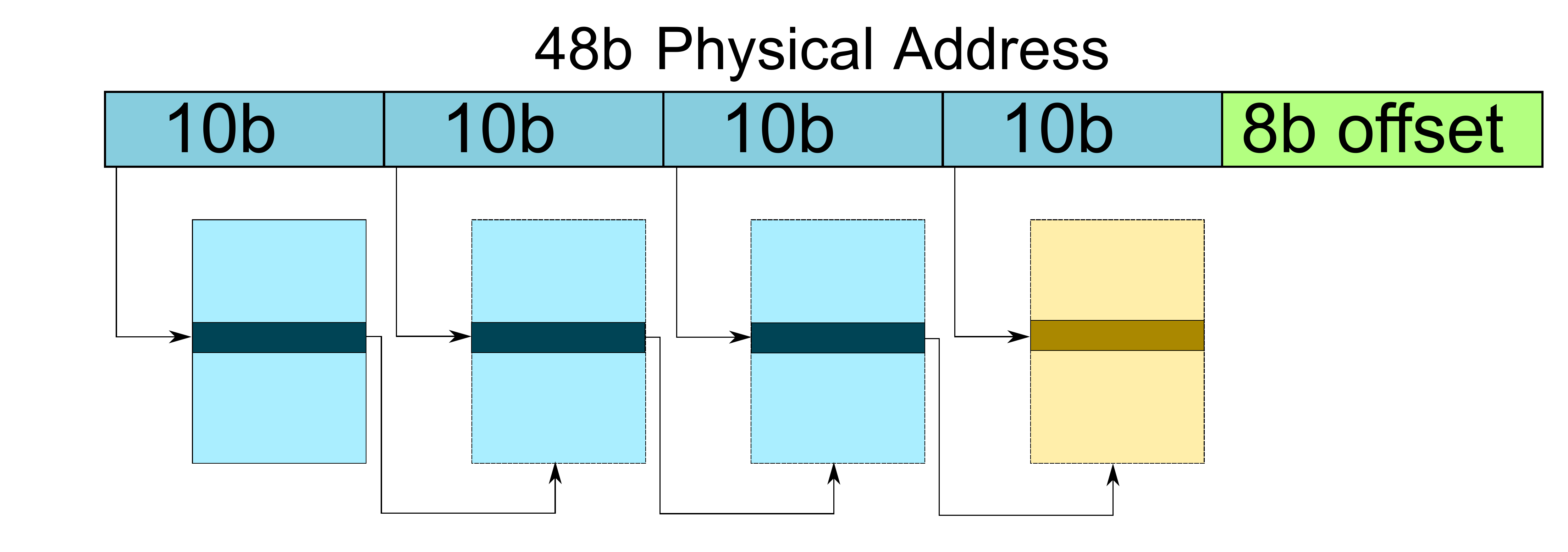}   
  \caption{Per-epoch DRAM page table.}
  \label{drampagetable}
\end{figure}

The design of the page table is shown in Figure \ref{drampagetable}. This table is indexed whenever we issue a write to the DRAM. Since writes are not on the critical path, this does not hurt the performance of the applications. We also iterate over the table when walking the DRAM. The least significant 8 bits are used as an offset. We implement our page table as a four-level tree, indexed using 10 bits at each level. At a non-leaf level (marked as blue in Figure \ref{drampagetable}), we store a physical address for the next page and a valid bit. The final level (marked as yellow in Figure \ref{drampagetable}) holds a {\em 2-bit private counter} for every DRAM page (to be discussed next) as well as a {\em valid} bit. It is the valid bits at the last level that signify whether a DRAM page belongs to an epoch's modified working set or not.

We use an incremental scheme to walk this page table. The idea is to balance the scrubbing traffic and regular post-epoch DRAM traffic.
Since we know the page size at each level, it is more efficient to have a per-level hierarchical counter on the memory controller (to maintain the progress of the page-walk process: one counter for each level of the page table). Alongside this counter, we cache each level's 48b starting page address. To walk the DRAM, we continually increment the lowest level counter to find a page to persist. If the resulting page has its valid bit set in the page table, it is read and persisted. This process continues until the last level counter reaches $2^{10}$ (the number of entries in each page at a level). When this happens, we calculate the entry's index corresponding to the last level of the page table by incrementing the counter stored in the previous level until we find a valid page. The page found, if not already cached (in a 3KB cache at the memory controller), is read from the DRAM, and the address stored in its 48b field is copied to the last-level counter as the base address for future walks. The process continues up the hierarchy until all valid pages have been walked. In this way, the hierarchical counter implements a nested \textbf{for} loop in hardware.  

\subsubsection{RQ2: Limiting the Number of Pages in DRAM}

It is necessary to bound the maximum number of pages to be walked in the DRAM to limit the latency of the checkpoints. Failing to do this results in an increase in the minimum time it takes to snapshot. Hence, we must periodically persist pages from the current epoch in the hope that a future snapshot will arrive sooner than a modification to those speculatively-persisted pages. This task calls for using a {\em locality predictor} alongside our DRAM scrubber.

We have not seen any prior work \cite{TSOPER,NVOverlay,WSP,THYNVM,PiCL} that treats a running epoch as a future to-be-persisted epoch. All previous works have assumed that the line between a running epoch and a persisting epoch is the epoch boundary. By persisting parts of the running epoch, we blur the epoch boundaries. However, the recovery epoch must remain consistent. Therefore, these speculated persists are not merged with the recovery page table until we reach the epoch boundary.

The upper bound that our locality predictor must work with is defined by the latency and is calculated by the checkpoint controller. Let $l$ be the latency value (in seconds) and $f$ be the frequency of the system. The number of cycles in which the snpashot must complete is $c = l  f$. If it takes $k$ cycles to persist a page, then the upper bound for an epoch is simply

\begin{equation} \label{eq.1}
n\le\frac{c}{k}
\end{equation}

However, this does not take into account the non-determinism on a real system and makes simplistic assumptions. Let us thus use this as a base and design a more sophisticated algorithm.

\subsubsection{Dynamic Tuning of the Locality Predictor}
It is important for the locality predictor to know when to presist a page. From Equation \ref{eq.1}, we know the value of $n$ that roughly corresponds to the target $CL$. If we were to keep no pages in DRAM that need to be persisted (similar to NVoverlay), we will achieve a checkpoint latency equivalent to the scrubbing time of caches. Let us call this the minimum latency $CL_{min}$. The acceptable observed latency ($CL_{obs}$) must therefore lie in the interval $(CL_{min},CL)$. 

The checkpoint controller, working with the memory controller tunes the locality predictor for optimal write-amplification. The idea is to achieve a $CL_{obs}$ close to $CL$. This means that the locality predictor should not be over-aggressive. Consequently, the system achieves an optimal write amplification for the given checkpoint latency.

To perform the tuning, an epoch is divided into 50 sub-epochs. At the start of each sub-epoch, the memory controller will calculate how close the epoch's modified set is to the theoretical maximum number of modified pages ($n$). Tracking any epoch's modified set requires a counter at the memory controller, which is incremented after a new page is added into the page table. Let the difference divided by $n$ be $\delta$. In this paper, we tune the aggressiveness of the predictor by changing its activation rate $\mathcal{R}$ (how frequently it runs). It can vary from 512-256K cycles. It varies linearly with $\delta$.

\subsubsection{Locality Prediction}
To make the prediction, we allocate a private 2-bit counter for each DRAM page and a shared 3-bit saturating counter for a group of 64 contiguous pages. The private counters are stored alongside the last level page table entry for an epoch, whereas the shared counters are stored separately. It can never be the case that a page is a part of both epochs' modified working set because we do not do any further address re-mapping/translation: we thus need to persist the earlier avatar of the page first. Therefore, we can safely store the private counter alongside the DRAM page-table entry. On the other hand, for shared counters, it is expected that a small subset of these counters will be used (due to locality) during a given time. As a result, separately storing these is a better idea. If these counters were present in the page table, a shared counter jump would mean a jump that is 64 pages away.

Every time we walk the DRAM page table, looking for a page to persist, we do a cyclic clear of bits in private and shared counters. This cyclic clear will clear bit 0 of the private counter in one operation and bit 1 in the next operation. The main idea is to capture both temporal and spatial locality. The shared counter captures spatial locality since nearby addresses share this counter. The cyclic clear captures temporal locality since it will eventually clear all bits of shared and private counters to zero for pages that have not been accessed in a long time. We can {\em persist} pages speculatively if we see that they are a part of the current epoch with their shared and private counters set to zero. 

%========================
%========================
%This comment is there to reduce the page size to <11

%When writing to a page, we set both the private and shared counters associated with the page to the maximum value possible. We do this instead of incrementing the counters by one to avoid the case where we increment a counter to 1 just before we clear its bit to 0. In this case, we mark the page as not having temporal locality, even when this is not the case.
%========================
%========================
\subsubsection{RQ3: Mispredictions and Data Coalescing}

To reduce write amplification, we {\em coalesce} data from the upper levels with data in the DRAM. To achieve this, the first time a request from the cache for a block reaches the access scheduler, a request is sent to the DRAM scrubber asking it to scrub the page that belongs to the cache line. The scrubber obliges if this page is present in the pre-snapshot epoch table. If the page is absent, the scrubber sends a \textit{nak} (negative acknowledgement) to the access scheduler.
  
The locality predictor may predict to persist a page before a request from the access scheduler arrives. This is the case of misprediction since a line from the cache has reached the access scheduler, whereas the predictor did not expect this. As a result, the misprediction penalty is that we send a \textit{nak} to the access scheduler, which must then persist the page again, creating wasted writes. Correctness is never violated since the
recovery table in the NVM is never updated speculatively. We don't describe the maintenance of this table and checkpoint recovery in detail because it is the same as that used by other competing work~\cite{NVOverlay}.

\subsection{Tunable Parameters}
Our system tunes the following values. The system determines these values using the $\langle ES, CL \rangle$ pair provided to the system. However, these parameters are not directly accessible to the user or the OS. This tuple is sent along with the token by the checkpoint controller.

\begin {enumerate}
\item \textit{scrubbing-step} : Frequency of the cache scrubber. 
\item \textit{scrubbing-granularity}: Number of rows (or sets) to scrub in one scrubbing-cycle.
\item \textit{memory-walk step}: Frequency of the DRAM walker.
\end{enumerate}

\section{Experimental Results}
\label{sec:results}
In this section, we implement and evaluate \nvper on Tejas \cite{tejas}, a cycle-approximate Intel PIN-based \cite{PIN} architectural simulator. We ran representative 7 benchmarks from the PARSEC \cite{PARSEC} benchmark suite. 

\subsection{Setup}

\begin{table}
  \centering
  \footnotesize
  \begin{tabular}{ |c|c| } 
    %\hline
    %\multicolumn{2}{|c|}{Simulator configuration}\\
    \hline
    Parameter & Value\\
    \toprule
    \hline
    Number of cores & 8\\
    Pipeline type & Out-of-order \\
    Frequency & 3200MHz \\
    Private Caches & L1 L2 \\
    Coherence & L2\\
    L1 Cache &  32K 4-cycle WT 8-way 64B block \\
    L2 Cache & 256K 8-cycle WB 8-way 64B block \\
        
    L3 Cache & 8-way Tiled SNUCA LLC \\
    L3 Cache Bank & 2M 30-cycle WB 8-way 64B block\\ 
    
    NoC Topology & TORUS\\
    DRAM ranks & 2 ranks, 8 banks per rank\\
    NVM ranks & 2 ranks, 8 banks per rank\\
    DRAM page size & 2Kb\\
    NVM page size  & 2Kb\\
    \hline
    
  \end{tabular}
  \caption{System setup. Our system setup is in line with recent work in the area of persistent memory architectures \cite{eightcores1,caches1}}
  \label{setup}
\end{table}

Table \ref{setup} show the system setup used for both the evaluation of NVoverlay and \nvper. Our L1 caches are write-through, and the L2 cache is a write-back coherent cache. This arrangement decreases the number of evictions from the coherent domain without L2 having to access L1 for each remotely-received request.
The shared LLC is a the L3 level, which is distributed all over the chip. We use the popular SNUCA~\cite{advarch} scheme. The system uses a torus NoC with eight cores and eight LLC tiles. The NoC also has nodes for the directory and memory controllers.  

We implement NVoverlay on Tejas. Parameters shown in Table \ref{setup} are used for NVoverlay as well with a few exceptions. In NVoverlay, the L1 cache is write-back and the L2 is shared by two cores. Unless stated otherwise, we equate NVoverlay's epoch time and \nvper $ES$ and set our $CL$ to 1. We  also implemented a version of \nvper called \nvper-\texttt{aggr}. In this version, data from LLC is directly persisted onto the NVM.

Since both NVoverlay and \nvper do not affect the critical path (loads do not get stalled), we see only a minor difference ($\approx 0.2\%$) in the performance between NVoverlay and \nvper. This minor difference is in favor of \nvper. Since NVoverlay has a write-back L1 cache (inherent requirement), any coherent request at the L2 cache must first check L1 caches of both cores for a newer version. Not doing so violates per-location sequential consistency \cite{advarch}.  

\subsection {Write Amplification}

\begin{figure}[htbp]
     \centering     
     \includegraphics[width = \columnwidth]{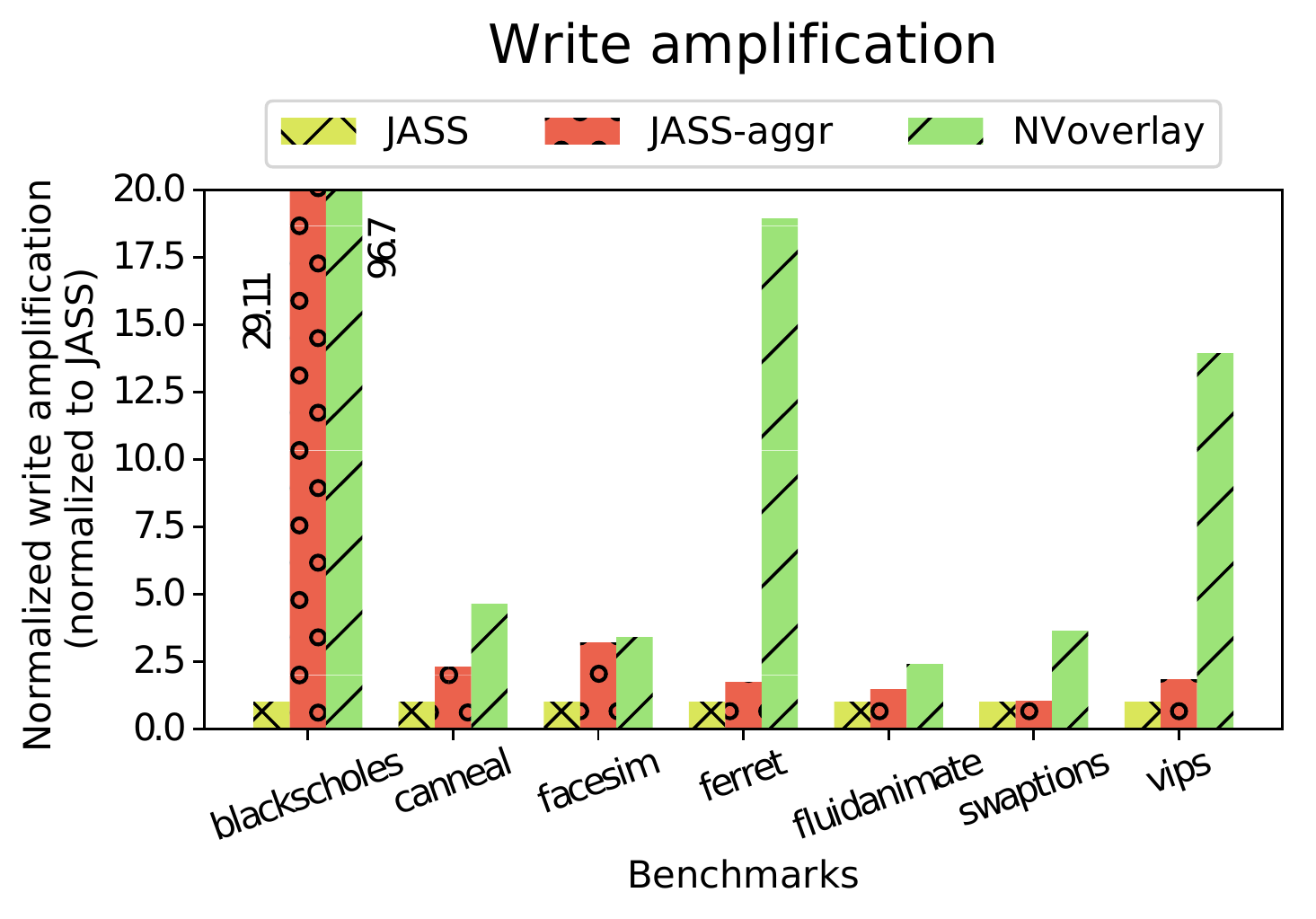}
     \caption{Normalized write amplifcation of \nvper and \nvper-\texttt{aggr} compared with NVoverlay.}
     \label{result.wafull}
\end{figure}

In this section, we compare the write amplification of \nvper with NVoverlay. Figure \ref{result.wafull} shows an improvement of at least 3$\times$ on most benchmarks. Moreover, \nvper-\texttt{aggr} performs better than NVoverlay in all benchmarks due to longer epochs as well as coalescing in the LLC. For some benchmarks like {\em blksch} the $WA$ improvement is as much as 96 $\times$.

The 3$\times$ improvement of \nvper over NVoverlay is largely attributed to the fact that data is not persisted on to the NVM upon cache eviction. Moreover, a 2$\times$ epoch size causes more coalescing to take place in the entire hierarchy (compared to only L1 and L2 for NVoverlay). However, we witness some outliers. For benchmarks with heavy communication  (\textit{ferret}), or small working sets (\textit{blksch,vips}), NVoverlay will create smaller epochs or perform less coalescing at the caches. \textit{Blksch} is the extreme outlier, since it creates $96\times$ more writes on NVoverlay compared to \nvper.

\subsection{Insights into the Operation of the System}

\begin{table*}
  \centering
  %\footnotesize
  \begin{tabular}{ |c|c|c|c|c|c|c|c|c| } 
    %\hline
    %\multicolumn{2}{|c|}{Simulator configuration}\\
    \hline
    {\bf Workload} & {\bf DRAM (ms)} & {\bf \% Cache} & {\bf \% DRAM}  & {\bf \% Coalescing} & {\bf LLC MPKI} & {\bf \% DRAM $WA$} \\    
    \hline
    \toprule
    \hline
    Blk.  & 5.25  & 88 & 12 & 13  & 0.20 & 50.16 \\
    Cann.  & 5.0  & 40 & 60 & 23  & 0.48 & 49.1  \\
    Face.  & 5.0  & 84 & 16 & 0.4 & 0.91  & 8.7   \\
    Ferr.  & 5.05  & 78 & 22 & 9.5 & 0.33 & 50.7  \\
    Fluid. & 5.0  & 47 & 53 & 1.6 & 0.51  & 10.7  \\
    Swap.  & 5.23  & 86 & 14 & 5.5 & 0.09 & 25.3   \\
    Vips   & 5.01  & 84 & 16 & 2.6 & 0.17 & 49.9  \\ 
    \hline
    
  \end{tabular}

  \caption{Breakdown of different aspects of the inspected system under the benchmarks. Time values for the longest epoch are reported for a latency constraint of 5ms. }
  \label{result.latency_breakup}
\end{table*}

Table \ref{result.latency_breakup} shows salient features of the execution: DRAM scrubbing time, \% of pre-snapshot data in the caches, DRAM, degree of coalescing, misprediction rate of the locality predictor, and the LLC MPKI (miss per kilo instructions). In this case the $CL=5$ ms. 
It is evident from column 2 in Table \ref{result.latency_breakup} that we are achieve our $CL$ target (within $\pm$ 5\%). Since DRAM and caches begin scrubbing independently, the slowest of the two will dictate the total time. Inevitably, DRAM scrubbing is the limiting factor. The caches are more or less fully full with pre-snapshot data. {\em Canneal} and {\em Fluidanimate} are notable exceptions given their access patterns. The DRAM (\%) numbers are relative: the numerator is the pre-snapshot footprint in the DRAM and the denominator is the total size of the snapshot. The degree of coalescing is small (0.4 to 23\%); most values are less than 10\%. 

The LLC MPKI gives us an insight to the locality of the application. Benchmarks with small ({\em blksch}) and medium ({\em swap,vips}) working sets have a smaller LLC MPKI due to increased locality. This distinction will help us understand the write amplification trends (in the next section).

\% DRAM $WA$ represents the additional writes that we need to perform owing to locality prediction
with reference to a scenario where we have a perfect predictor. We see an additional
8.7-50.1\% writes. This leads to write amplification.

We also tune our parameters dynamically. In our setup, we set the inital values of \textit{scrubbing-step} and \textit{memory-walk step} to 1K cycles each and {\em scrubbing-granularity} is set to 1 (one cache row at a time). We allow the {\em memory-walk step} to be tuned within the range $(0.5K,256K)$ cycles. These values represent a vast tuning range for our experiments. Since cache scrubbing is faster than DRAM scrubbing, we have more opportunities to steal a cache's cycle or even just a read port to read entire rows in a very fast way. This causes the caches to rarely tune their {\em scrubbing-step} values.

\subsection {Sensitivity to Input Parameters}

\begin{figure}[htbp]
     \centering     
     \includegraphics[width = \columnwidth]{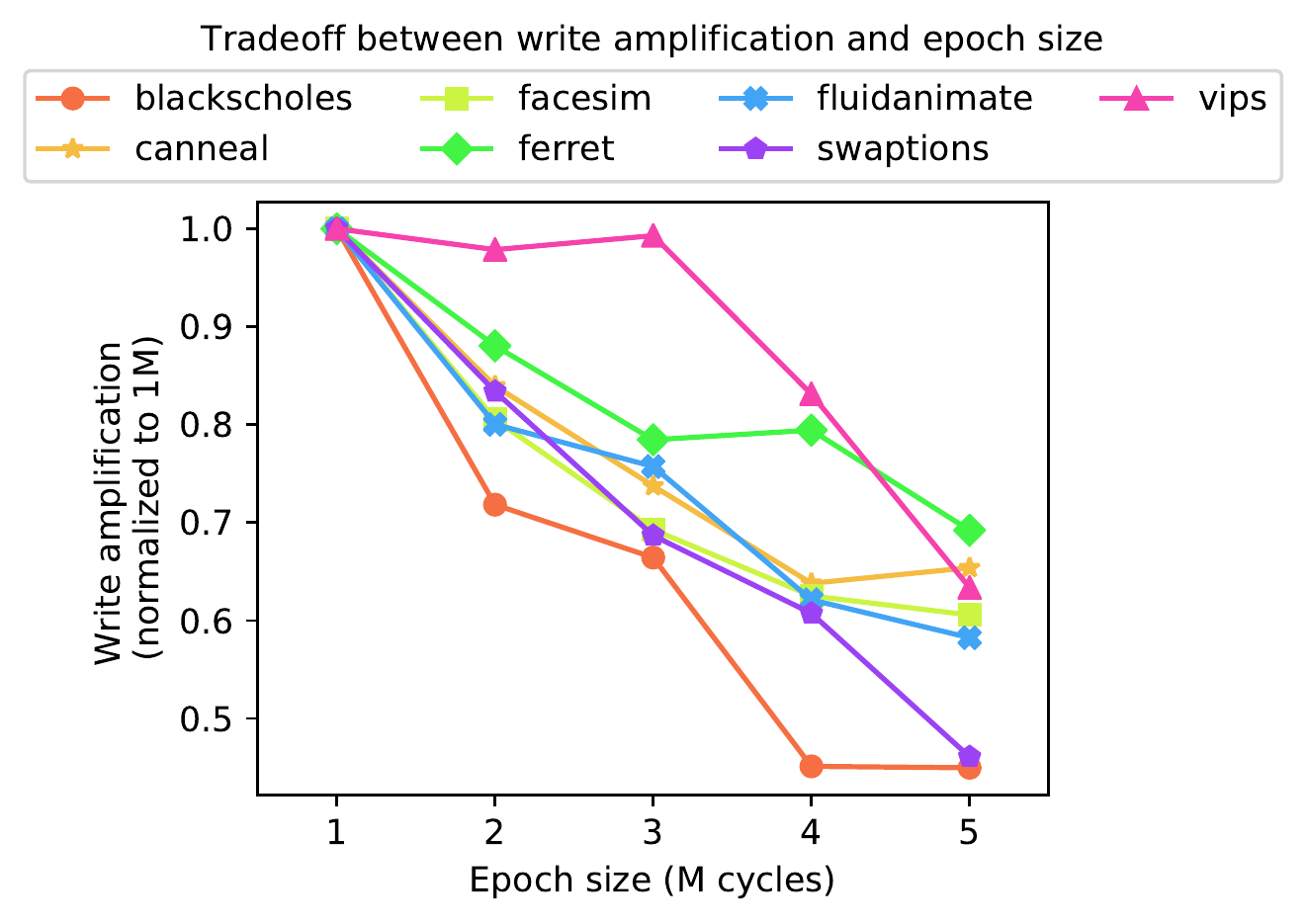}
     \caption{Normalized write amplifcation of \nvper with increasing epoch size ( compared with NVoverlay).}
     \label{result.epoch_wa}
\end{figure}

In this section we run two experiments to compare \nvper to NVoverlay. In the first experiment, we run \nvper with five different epoch sizes. In the second experiment, we run \nvper on five different latency constraints. We discuss the results below.

\subsubsection{Epoch Size ($ES$)}

We evaluate \nvper with five different epoch sizes and study the effects of write amplification with increasing $ES$.  $CL$ is set to 1 ms. The results in Figure \ref{result.epoch_wa} show that $WA$ is reduced by a factor of 0.8$\times$ when we increase $ES$ from 1M to 2M. The decrease in $WA$ is attributed to the coalescing performed by the hierarchy. Diminishing returns set in after a point (as expected), which is different for every benchmark.

Benchmarks with small working sets (\textit{blksch}) are expected to see the curve flatten after a point. On the other hand, benchmarks like (\textit{vips}) have access patterns where the increasing epoch size is beneficial after a certain point.

\subsubsection{Checkpoint Latency ($CL$)}

\begin{figure}[htbp]
     \centering     
     \includegraphics[width = 0.85\columnwidth]{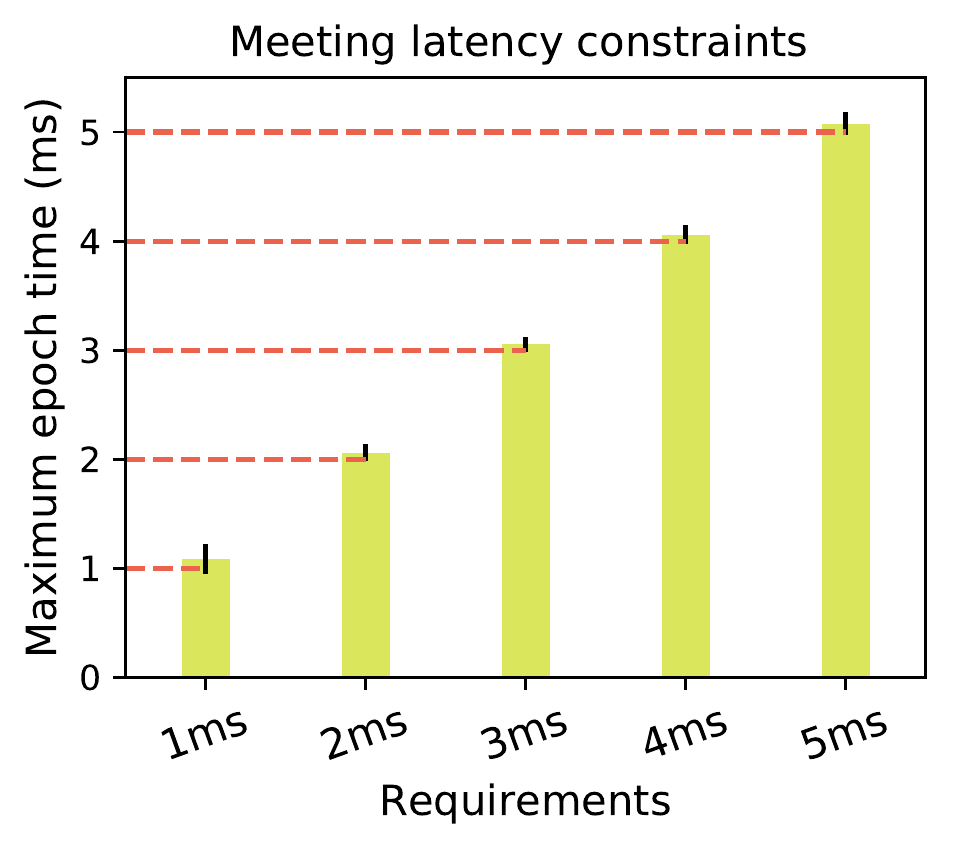}     
     \caption{Maximum checkpoint completion time for \nvper under different checkpoint latency constraints.}
     \label{result.latency}
\end{figure}

\begin{figure}[htbp]
     \centering     
     \includegraphics[width = \columnwidth]{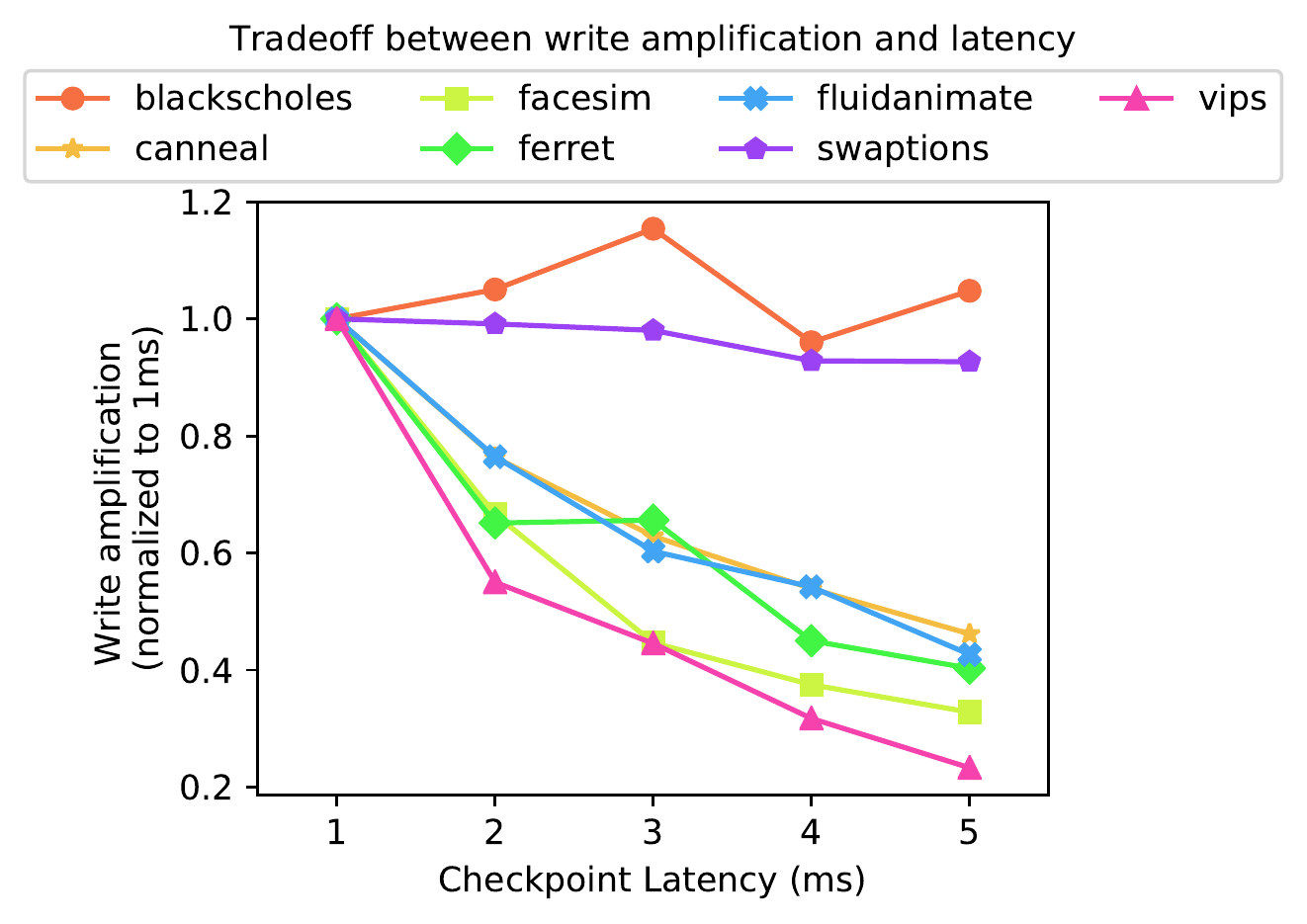}
     \caption{Tradeoff between latency and write amplification.}
     \label{result.clear_wa_cl}
\end{figure}

We compare \nvper with five different $CL$ values: 1, 2, 3, 4, and 5 ms. We show that we always meet the $CL$ constraint (refer to Figure \ref{result.latency}).  We see that given a latency, \nvper meets it. $CL_{obs}$ (observed $CL$) is always within $5\%$ of  the target $CL$.

The corresponding write amplification is shown in Figure \ref{result.clear_wa_cl}. This figure illustrates the tradeoff between write amplification and the checkpoint latency. The graph has a negative slope, i.e., with increasing latency, write amplification decreases. \nvper is unique in this space. Our system, given a $CL$, will minimize the write amplification as well as meet the checkpoint latency.

From Table \ref{setup}, we can see that benchmarks with relatively small working sets (\textit{blksch,swaptions}) scale less effectively with increasing latency than the others in Figure \ref{result.clear_wa_cl}. Since the value of $CL$ controls the locality predictor (Equation \ref{eq.1}), increasing $CL$ beyond the maximum limit of the system does not see any benefits. The maximum limit is a function of the working set. On the other hand, systems with a larger memory footprint  see a steeper decrease in write amplification with an increasing $CL$ target up till a certain point (as expected).

\section{Related Work}
\label{sec:rel}

In this section, we look at previous system level checkpointing proposals: ThyNVM~\cite{THYNVM} and PiCL~\cite{PiCL} first. These works directly motivated our closed competing work, NVoverlay \cite{NVOverlay}.

\subsection{ThyNVM and PiCL}
ThyNVM~\cite{THYNVM} divides each program's execution into time-ordered epochs. Each epoch has two phases: an execution phase and a checkpointing phase. To avoid application stalls, the checkpointing phase of an epoch runs concurrently with the execution phase of the next epoch. ThyNVM incurs a large write amplification since it maintains three versions of data at any given time (current,last,penultimate). All three copies are required since a failure can corrupt both the working copy and a checkpoint in progress.

PiCL~\cite{PiCL} aims to improve on ThyNVM by reducing the epoch sizes further. It uses undo logging in the caches. Although not a problem at the level of blocks (since logs are persisted in batches), the system will produce a log for every write operation on the cache. Such a scheme is prohibitive since the baseline write amplification is high because of the writes to the log.

\subsection{NVoverlay}
NVoverlay \cite{NVOverlay} aims to improve upon both PiCL \cite{PiCL} and ThyNVM \cite{THYNVM} by further reducing the size of each epoch. NVoverlay defines epoch boundaries by monitoring coherence traffic, thus creating prohibitively small epochs. The system uses a distributed vector clock with no global synchronization. This causes some cores to run ahead of others (in terms of the epoch number). 

Designed for high-frequency checkpointing, NVoverlay tags each cache line with a 16-bit epoch id. Moreover, the system tags each DRAM page using ECC space, which is not practical. Multiple versions of a cache line can exist in the hierarchy. The resulting cache pollution is handled by the scrubber.  

The {\em recovery epoch} is the latest epoch that was persisted on all cores. A multi-version snapshotting mechanism manages different versions of page tables  for each epoch in DRAM, while a single recovery epoch's page table is present in the NVM. To reduce writes to NVM, NVoverlay unmaps data in the NVM at cache line granularity while mapping is done at page granularity. This causes the system to occasionally perform garbage collection of sparsely mapped pages.  All of these have high overheads and as we have argued, such high-frequency checkpointing at the expense of a high $WA$ is counter-productive.

\section{Concluding Remarks}
\label{sec:conc}

In this paper, we proposed a very different paradigm for architecting NVM-based checkpointing systems.
Our key insight was that NVM-based systems should typically not be designed for instant recovery. Most applications
do not have such requirements, and the penalty paid in terms of write amplification is very high. Instead,
we follow a more nuanced approach. We allow the user to set a checkpointing frequency and checkpoint latency
as per her requirements. We guarantee that the checkpoint latency is met and furthermore the write amplification
is minimized. To the best of our knowledge, no other competing work provides such a facility where the 
behavior of the system can be tuned as per the application user's requirements. Our design minimizes the $WA$
at any operating point on a best-effort basis. We show that as compared to NVOverlay \nvper reduces the $WA$
by 2.3-96\%.

\newpage

\bibliographystyle{plain}
\bibliography{refs.bib}

\clearpage
\appendix

\section{Proof of Termination Detection on NoC}
\label{appendix.proof}
%\begin{theorem}
%  Algorithm \ref{flushalgo} will correctly flush all pre-snapshot messages in the system. 
%\end{theorem}

%\begin{proof}

\textbf{Terminology} : The set of connected routers to a router is represented by $neighbors(R_i)$. Each router has an internal bit that states whether it knows of an ongoing checkpoint or not. Let us denote this bit by $state(R_i)$.

Algorithm \ref{flushalgo} describes the operation of routers in the presence of an ongoing snapshot. $R_i$ describes the $i^{th}$ router. $buffer$ represents the internal buffers of a router. Each router can be in one of two states: token received ($TR$), or no token received ($NTR$). Normal operations continue when the router is in state $NTR$, and an incoming message is not a token.

\begin{lemma}
  \label{LEMMA1}
  A post-snapshot message never reaches a router in state $NTR$.  
\end{lemma}

\begin{proof}
  
  Let us assume that at router $R_i$, we have received a post-snapshot marked message while $state(R_i)=NTR$. Since the token has the highest priority in our network and is propagated on all links, this is not possible. The token would have arrived first.    
\end{proof}

\begin{lemma}
  \label{LEMMA2}
  All pre-snapshot marked messages, which are not being currently transmitted on a link, are counted by at most a single router.  
\end{lemma}

\begin{proof}
  Assume a message is counted by two routers: $R_i$ and $R_j$. Assume $R_i$ counted it first. This means that after it counted the message, it must have sent the token in $R_j$'s direction. The message could only have been sent after the token. Given that the token has the highest priority, the message will never be able to overtake it. This means that when the message reaches $R_j$, $R_j$ is already in a post-snapshot state. This further means that $R_j$ cannot {\em count} the message. This leads to a contradiction. 
\end{proof}

We may optimize the marking of our message even further. Since the token is guaranteed to travel first across all links, it can never be the case that an incoming pre-snapshot message on a token-received link was not marked. However, there can be a case of messages transmitted on a link that may be missed. We fix this next.  

\textbf{Missed messages}: Suppose at router $R_i$, one of its neighbors $R_j$ in state $NTR$ sends a message to $R_i$. Let's say the token is in transit from $R_i$ to $R_j$ at the same time. As a result, the sender $R_j$ does not count the message. We must therefore count it at the receiver $R_i$. Consequently, we propose that each router starts a small state machine on every link after changing its state from $NTR$ to $TR$. This state machine counts down to the delay of the associated link. The router must count any messages that may have possibly been in transit as long as they appear within a window of sending the token to a neighbor (when the counter is non-zero), and it is a pre-snapshot message. Once the state machine has counted down to zero, $R_i$ can correctly assume that the receiver will count all pre-snapshot messages on its end. There are no messages in-flight on any link that are uncounted.
  
\begin{lemma}
  \label{LEMMA3}
  Employing the \textbf{missed messages} scheme (as discussed above) will count all in-flight messages at the destination.
\end {lemma}

\begin{proof}
  Two scenarios can break this hypothesis. The first is an uncounted message, and the second is the double counting of a message. We show that both these cases are impossible.
  
  An uncounted message must be transmitted after the token is sent on a link; this is to avoid counting on the receiving side. On the other hand, for a sender to not count a message, it must be transmitted before the token. Since links are FIFO, we have a contradiction.
  
   On the other hand, for a message to be counted twice, it needs to start its journey on a link $L$ after the sender has received the token but arrive before the token has been received at the destination router. This case can never happen since individual links are FIFO in an NoC.
\end{proof}

A subtle problem with termination detection can occur due to a coherence directory. It may be the case that a pre-snapshot message to a directory generates new messages in the NoC. Since a directory is not a processing element for some messages, but rather a concierge where messages go to know their final destination, there may be a case that the directory generates more than one message for every pre-snapshot marked message it receives.
  
\textbf{Updating xcount}: The system must be made aware of new pre-snapshot messages in the NoC. To do so, the router at the directory is specialized. It does not update its $pcount$ until it receives a signal from the directory. The directory will signal its router if the pre-snapshot message has been eliminated from the network and there are no other messages generated. If the directory generates exactly one message, the router is notified about this, and neither the $xcount$ or $pcount$ changes (one message in, one message out). If the directory generates $n \ge 2 $ messages, 
$xcount$ needs to be incremented by $n-1$ (the initial message is consumed). This is done by notifying the checkpoint controller.

\begin{theorem}
  Using this scheme for \textbf{updating xcount} and \textbf{pcount} (as discussed above), no pre-snapshot message remains in the NoC after a flush termination.
\end{theorem}

\begin{proof}
  
      Flush termination is detected when $pcount = xcount$, and all the routers have reported their values. Let us assume, at this point, there is a message somewhere in the NoC that is still marked pre-snapshot. $pcount$ will always be the number of messages that have successfully reached their destination. Since $pcount = xcount$ is true; it is the case that $xcount$ was miscounted. Using Lemmas \ref{LEMMA2} and \ref{LEMMA3}, we know that $xcount$ has counted all pre-snapshot messages at the routers correctly. Using the \textbf{updating xcount} scheme, we know that newly created pre-snapshot messages have been counted by the directory's router and updated at the checkpoint controller. Hence $xcount$, has been counted correctly, and our initial assumption was wrong. Thus, proved by contradiction.
  \end{proof}

\end{document}